\documentclass{article}
\usepackage{spconf,amsmath,graphicx,hyperref}
\usepackage{cite}
\usepackage{amsthm}
\usepackage{amsmath,amssymb,amsfonts}
\usepackage{algorithmic}
\usepackage{graphicx}
\usepackage{textcomp}
\usepackage{xcolor}
\usepackage{cleveref}
\usepackage{enumitem}
\usepackage{etoolbox}
\usepackage{graphicx}
\usepackage{makecell}
\usepackage{caption}
\usepackage{wrapfig}
\usepackage{setspace}
\usepackage{subcaption}
\usepackage{floatflt}
\usepackage{float}
\usepackage{setspace}
\usepackage{tikz}

\renewrobustcmd{\bfseries}{\fontseries{b}\selectfont}
\renewrobustcmd{\boldmath}{}
\newrobustcmd{\B}{\bfseries}

\newsavebox\CBox

\makeatletter

\renewcommand{\section}{\@startsection
   {section}%
   {1}%
   {}%
   {-0.2\baselineskip}%
   {0.2\baselineskip}%
   {}}%

\renewcommand{\subsection}{\@startsection
  {subsection}%
  {2}%
  {}%
  {-0.2\baselineskip}%
  {0.1\baselineskip}%
  {}}%

\renewcommand{\subsubsection}{\@startsection
  {subsubsection}%
  {3}%
  {}%
  {-0.2\baselineskip}%
  {0.1\baselineskip}%
  {}}%
\newcommand{\unnumberedsection}[1]{
    \@startsection{section}{1}{2.7cm}{-0.1\baselineskip}{0.1\baselineskip}{\normalfont\footnotesize\bfseries}*{#1}
}

\g@addto@macro\normalsize{%
  \setlength\abovedisplayskip{3pt plus 2pt minus 1pt}
  \setlength\belowdisplayskip{3pt plus 2pt minus 1pt}
  \setlength\abovedisplayshortskip{2pt plus 2pt minus 1pt}
  \setlength\belowdisplayshortskip{2pt plus 2pt minus 1pt}
}

\newtheorem{theorem}{Theorem}
\newtheorem{lemma}{Lemma}

\DeclareMathOperator*{\argmax}{argmax}

\title{Sequence-Level Unsupervised Training in Speech Recognition: A Theoretical Study}
\name{Zijian Yang$^1$, Jörg Barkoczi$^1$, Ralf Schlüter$^{1,2}$, Hermann Ney$^{1,2}$}
\address{$^1$Human Language Technology and Pattern Recognition, Computer Science Department\\
RWTH Aachen University, Germany\\
$^2$AppTek GmbH, Germany}
\begin{document}
%
\maketitle

\begin{abstract}
Unsupervised speech recognition is a task of training a speech recognition model with unpaired data. To determine when and how unsupervised speech recognition can succeed, and how classification error relates to candidate training objectives, we develop a theoretical framework for unsupervised speech recognition grounded in classification error bounds. We introduce two conditions under which unsupervised speech recognition is possible. The necessity of these conditions are also discussed. Under these conditions, we derive a classification error bound for unsupervised speech recognition and validate this bound in simulations. Motivated by this bound, we propose a single-stage sequence-level cross-entropy loss for unsupervised speech recognition.
\end{abstract}
\section{Introduction \& Related Work}
\label{sec:intro}
Unsupervised speech recognition is a task to train an automatic speech recognition (ASR) model with unpaired speech and text data. Since training requires no annotated data, it is well suited to low-resource languages with limited transcriptions. Recent works \cite{liu18g_interspeech,yeh2019unsupervised, chen19e_interspeech, baevski2021unsupervised, liu2023towards} have shown great success in unsupervised speech recognition. 

Most previous works \cite{chen19e_interspeech,baevski2021unsupervised, liu2023towards} use a GAN-based criterion \cite{goodfellow2020generative}. \cite{wang2024unsupervised} proposed a training criterion based on the $\ell_1$ distance. However, these criteria are defined under a deterministic mapping from speech units to phoneme labels, whereas modern ASR systems are inherently statistical models. Consequently, these methods typically follow a two-stage pipeline: (i) unsupervised training to obtain an initial mapping function, and (ii) semi-supervised training of a standard ASR model (e.g., CTC\cite{graves2006connectionist}) on pseudo-labels generated by that mapping. However, whether a unified single-stage training criterion for unsupervised speech recognition exists for statistical models remains an open question.

The theoretical background also requires further investigation. In general, it is impossible to obtain the joint distribution from marginal distributions without any constraints/assumptions. \cite{wang-etal-2023-theory} built up a theoretical framework for GAN-based approaches and discussed sufficient conditions that allow unsupervised speech recognition. They also derived conditions guaranteeing the convergence to a global optimum. However, the behavior of the model when training does not reach the global optimum remains unclear. In particular, the relationship between the training loss and the sequence classification error is not established. 

To address these limitations, we introduce a theoretical framework for sequence-level unsupervised training based on classification error bounds \cite{ney2003relationship,nussbaum2013relative,yang2025classification}. In this work, we:

\begin{itemize}
\setlength\itemsep{0.1em}
    \item consider a statistical model rather than a deterministic function;
    \item introduce and discuss conditions which enables unsupervised speech recognition
    \item derive a classification error bound for unsupervised training in speech recognition
    \item propose a sequence-level cross-entropy loss for statistical models based on the theoretical bound 
\end{itemize}

\section{Classification Error Mismatch in ASR}

For statistical sequence classification tasks such as ASR, classification error is an important performance measure \cite{ney2003relationship,nussbaum2013relative, yang2025classification}. Let $x_1^N \in \mathcal{X}^N$ be an observation sequence with length $N$, where $\mathcal{X}$ is a set of speech units, and $c_1^N \in \mathcal{C}^N$ be a label sequence, where $\mathcal{C}$ is a set of discrete labels. In this work, we consider discrete speech units and assume $|\mathcal{X}| > |\mathcal{C}|$. Nevertheless, main conclusions should extend to continuous speech units by carefully replacing summations with integrals and probability masses with probability densities. In ASR, the recognition/decoding result is generated by using a decision rule, which is a mapping function from the observation sequence $x_1^N$ to the label sequence $c_1^N$. By utilizing the underlying true distribution, Bayes decision rule minimizes the expected error given an error function $L[c_1^N, \Tilde{c}_1^N]$. To simplify the discussion, we ignore the alignment problem between sequences and assume that the observation and label sequences have the same length $N$. Following \cite{yang2025classification}, we apply the averaged Hamming distance $L[c_1^N, \Tilde{c}_1^N] := \frac{1}{N}\sum_{n=1}^N [1 - \delta(c_n, \Tilde{c}_n)]$ as the error function between two label sequences, where $\delta$ is Kronecker-Delta. Let $pr(x_1^N, c_1^N)$ be the true joint distribution of sequence pairs $(x_1^N, c_1^N)$, the Bayes decision rule is defined as follows \cite{yang2025classification}:
\begin{equation*}
    c_{*,1}^N(x_1^N) =\left( \big(c_n^*(x_1^N)\big)_{n=1}^N| c_n^*(x_1^N) = \argmax_c pr_n(c,x_1^N) \right),
\end{equation*}
where $pr_n(c,x_1^N) := \sum_{c_1^N: c_n=c} pr(c_1^N, x_1^N)$.
However, in practice, the underlying true distribution $pr$ used in Bayes decision rule is not available. Therefore, we train a model to estimate the true distribution on the training data, and replace the true distribution $pr$ with the model distribution $q$ in the decision rule. The model-based decision rule is defined as:
\begin{equation*}
    c_{q,1}^N = \left(\big(c_n^q(x_1^N)\big)_{n=1}^N|c_n^q(x_1^N)= \argmax_c q_n(c,x_1^N)\right),
\end{equation*}
where $q_n(c,x_1^N): = \sum_{c_1^N: c_n=c} q(c_1^N, x_1^N)$. The mismatch between the true distribution $pr$ and the model distribution $q$ leads to a corresponding mismatch in classification errors between the Bayes decision rule and the model-based decision rule. As discussed in \cite{yang2025classification}, the classification error mismatch for sequences is defined as:\\
\scalebox{0.94}{\parbox{1.06\linewidth}{
\begin{equation*}
        \overline{\Delta_q} = \frac{1}{N} \sum_n \Delta_q^n
\end{equation*}
}}
where\\
\scalebox{0.94}{\parbox{1.06\linewidth}{
\[\Delta_q^n = \sum_{x_1^N} \Big [pr_n\big(c_n^*(x_1^N), x_1^N\big) - pr_n\big(c_n^q(x_1^N),x_1^N\big) \Big]\] 
}}
is the local error mismatch. By definition, $\overline{\Delta}_q \geq 0$.
\cite{nussbaum2013relative, yang2025classification} have shown that in the supervised training setting, where the joint distribution $pr(c_1^N, x_1^N)$ is accessible in the training data, $\overline{\Delta_q}$ can be bounded by the KL divergence between $pr(c_1^N, x_1^N)$ and $q(c_1^N, x_1^N)$. This result provides a theoretical justification for the use of cross-entropy loss.

In contrast, under the unsupervised training scenario, the joint distribution is not available. Consequently, it remains unclear whether (and how) one can define an appropriate training criterion to minimize $\overline{\Delta}_q$.
\section{Unsupervised Speech Recognition}
\subsection{Problem Statement}
\label{sec:problem}
For unsupervised training, instead of the joint distribution $pr(x_1^N, c_1^N)$, only the marginal distributions $pr(x_1^N)$ and $pr(c_1^N)$ are available. Unlike previous works \cite{baevski2021unsupervised, wang-etal-2023-theory, wang2024unsupervised, liu2023towards}, we consider the conditional distribution $q(x | c)$ of a generative model. This modeling approach has two potential advantages: it is more natural to sum over $x_1^N$ samples from the training data while modeling the language model probability $pr(c_1^N)$ in accordance with standard ASR training practices; and it facilitates the transfer of theoretical results from the discrete case of $x$ to the continuous case when working with $q(x|c)$.
The model conditional distribution is defined as follows:\\
\scalebox{0.94}{\parbox{1.06\linewidth}{
\begin{align*}
    q(x_1^N|c_1^N) := \prod_{n=1}^N q(x_n|c_n)
\end{align*}
}}
Furthermore, we assume the availability of sufficient labeled data such that the label sequence prior can be modeled exactly, i.e., $q(c_1^N) = pr(c_1^N)$. The model marginal distribution of observation sequences is defined as:\\
\scalebox{0.94}{\parbox{1.06\linewidth}{
\begin{align*}
    q(x_1^N) = \sum_{c_1^N} pr(c_1^N) q(x_1^N|c_1^N)
\end{align*}
}}
In general, it is impossible to upper-bound $\overline{\Delta}_q$ by the difference between $pr(x_1^N)$ and $q(x_1^N)$ without further conditions. A simple case is to have $N=1$, which is refered to as single event case in the following discussions. We can then omit the position index $n$. Consider the exact-match case, i.e. $q(x) = pr(x)$. Given the true distributions, there are in total $|\mathcal{X}|$ linear equations for $|\mathcal{X}| \cdot |\mathcal{C}|$ variables $\{q(x|c)\}$. In general, there is no unique solution of $q(x|c)$. Moreover, it is possible to come up with $pr$ and $q$ that $pr(x)=q(x)$ while $\overline{\Delta}_q>0$ .
Consequently, we must specify conditions that guarantee the possibility of non-trivial proper solutions to unsupervised speech recognition.

Since $\overline{\Delta_q}$ is not straightforward to compute, we apply the $\ell_1$ distance bound of $\overline{\Delta_q}$ derived in \cite{ney2003relationship}:\\
\scalebox{0.93}{\parbox{1.07\linewidth}{
\begin{align}
\label{eq:basic}
    \overline{\Delta_q} \leq \frac{1}{N}\sum_{n=1}^N \sum_{x_1^N}\sum_c |pr_n(c,x_1^N) - q_n(c,x_1^N)| :=\overline{D}_q,
\end{align}
}}
and consider $\overline{D}_q$ instead of $\overline{\Delta}_q$. 
In the next section, we will introduce two conditions that can make $\overline{D}_q$ bounded by the $\ell_1$ distance between $pr(x_1^N)$ and $q(x_1^N)$.

\subsection{Sufficient Conditions for Unsupervised Training}
Firstly, for the model to perfectly capture the true distribution, both must share the same structural form. This leads to our first condition: the true distribution should have the same decomposition as the model distribution, i.e., the true conditional distribution can be factorized in the same way as the model distribution. This condition is referred to as the structure constraint in the following discussions.\\
\scalebox{0.94}{\parbox{1.06\linewidth}{
    \begin{align*}
        pr(x_1^N|c_1^N) = \prod_{n=1}^N pr(x_n|c_n) 
    \end{align*}
}}
\cite{baevski2021unsupervised, liu2023towards, wang2024unsupervised} achieved successful unsupervised training with a localized mapping from segment representations to phonemes. Therefore, we assume this decomposition with local dependencies is valid in practice.\\
Secondly, the labels must be mutually distinguishable. For instance, if replacing one label $c$ with another label $c'$ leaves the label sequence prior probability unchanged, then these two labels cannot be distinguished from the marginal distribution, making it impossible to recover the underlying joint distribution. This leads to our second condition: the language model matrix\\
\scalebox{0.94}{\parbox{1.06\linewidth}{
\[\mathbf{P}_C \in \mathbb{R}_{\geq0}^{N \times |\mathcal{C}|}: \left (\mathbf{P}_C\right )_{n,c} = pr_n(c) :=\sum_{c_1^N: c_n=c}pr(c_1^N)\]
}}
should have full column rank, where $pr_n(c)$ is the marginal probability of having label $c$ at position $n$. This constraint indicates that the labels cannot replace each other, even not possible for linear combinations, in terms of position-dependent unigram probabilities. We computed the smallest singular value $\sigma_{\min}$ of $\mathbf{P}_C$ on LibriSpeech transcriptions and obtained
$\sigma_{\min}\approx 3\times 10^{-4}$.
While small, this value is clearly non-zero, suggesting that $\mathbf{P}_C$ is numerically full rank on the evaluated data.
Moreover, as the amount of data increases,
$\sigma_{\min}(\mathbf{P}_C)$ typically increases.
We therefore believe this assumption is likely to hold in practice.

Under these conditions, we obtain the following bound:
\begin{theorem}

\label{theorem:seqineq}
When $\mathbf{P}_C$ has full column rank, and the true distribution satisfies the structure assumption, the following inequality holds:
    \begin{align}
     \label{eq:global_seq}
 \overline{D}_q\leq N^2\|P_C^{+}\|_1 \sum_{x_1^N} |pr(x_1^N) - q(x_1^N)|, 
    \end{align}
where $\mathbf{P}_C^+ = ({\mathbf{P}^\top_C}\mathbf{P}_C)^{-1}\mathbf{P}^\top_C$ is the left-inverse of $\mathbf{P}_C$, and $ \|\mathbf{P}_C^{+}\|_1 $ is the induced $\ell_1$ norm of $\mathbf{P}_C^{+}$.
\end{theorem}
In the following, we will derive the proof of Theorem \ref{theorem:seqineq} with two lemmas. 
We first compare the local distributions by introducing the following $\ell_p$ distances:\\
\scalebox{0.93}{\parbox{1.07\linewidth}{
\begin{align*}
    d^\text{cond}_p(pr, q) &: = \big (\sum_{x\in \mathcal{X}}\sum_{c\in \mathcal{C}} |pr(x|c) -q(x|c)|^p \big)^\frac{1}{p}\\
    d_p^{\text{pos}}(pr,q) &:= \big( \sum_{n=1}^N\sum_{x\in \mathcal{X}} |pr_n(x) - q_n(x)|^p \big )^\frac{1}{p}\\
\end{align*}
}}
where $pr_n(x) =\sum_{x_1^N: x_n =x} pr(x_1^N)$ is the marginal position-dependent unigram probability. $q_n(x)$ is similarly computed by $q(x_1^N)$.
Based on the structure constraint, we have:\\
\scalebox{0.93}{\parbox{1.07\linewidth}{
\begin{align*}
     \sum_c pr_n(c) \big( pr(x|c) -q(x|c)\big) =  pr_n(x) - q_n(x)
\end{align*}
}}
Rewrite the above equation in matrix form, we have:
\begin{align}
\label{eq:loca}
    \mathbf{P}_C \big( \mathbf{pr}^C_{x} - \mathbf{q}^C_{x}\big) = \mathbf{pr}^N_{x} - \mathbf{q}^N_{x},
\end{align}
where $\mathbf{pr}^C_{x} \in \mathbb{R}_{\geq0}^{|\mathcal{C}|}$ is a vector, with $\left(\mathbf{pr}^C_{x} \right)_c = pr(x|c)$, and $\mathbf{pr}^N_{x} \in \mathbb{R}_{\geq0}^{N}$ is a vector, with $\left( \mathbf{pr}^N_{x}\right )_n = pr_n(x)$. $\mathbf{q}^C_{x}$ and $\mathbf{q}^N_{x}$ are similarly defined for model distribution $q$. With the expression in the matrix form, we have:
\begin{lemma}
\label{lemma:local}
When $\mathbf{P}_C$ is full column rank, the $l_p-$distance between $\{pr(x|c)\}$ and $\{q(x|c)\}$ is bounded by the following inequality:\\
\scalebox{0.94}{\parbox{1.06\linewidth}{
    \begin{align*}
      d^\text{cond}_p(pr, q) \leq \|\mathbf{P}_C^+\|_p  \cdot  d^\text{pos}_p(pr, q) \\
    \end{align*}
    }}
    \scalebox{0.94}{\parbox{1.06\linewidth}{
    where
    \begin{align*}
        \|\mathbf{P}_C^{+}\|_p = \sup_{\mathbf{v}\neq 0}\frac{\|\mathbf{P}_C^{+}\mathbf{v}\|_p}{\|\mathbf{v}\|_p}
    \end{align*}
    \vspace{-0.1em}
}}
\end{lemma}\noindent
\textbf{Proof of Lemma \ref{lemma:local}}\\
    When $\mathbf{P}_C$ is column full rank, the left-inverse $\mathbf{P}^+_C$ exists. More specifically, we have:\\
    \scalebox{0.94}{\parbox{1.06\linewidth}{
    \begin{align*}
        \mathbf{pr}^C_{x} - \mathbf{q}^C_{x} = \mathbf{P}_C^+ (\mathbf{pr}^N_{x} - \mathbf{q}^N_{x})
    \end{align*}
    }}
    By computing the $l_p-$norm on both sides, we obtain:\\
    \allowdisplaybreaks
    \scalebox{0.93}{\parbox{1.07\linewidth}{
    \begin{align*}
       & \| \mathbf{pr}^C_{x} - \mathbf{q}^C_{x} \|^p_p = \| \mathbf{P}_C^+ (\mathbf{pr}^N_{x} - \mathbf{q}^N_{x}) \|^p_p
        \leq \| \mathbf{P}_C^+ \|^p_p   \| \mathbf{pr}^N_x - \mathbf{q}^N_x \|_p^p\\
        &\Rightarrow \sum_x  \| \mathbf{pr}^C_{x} - \mathbf{q}^C_{x} \|^p_p \leq  \| \mathbf{P}_C^+ \|^p_p  \sum_x \| \mathbf{pr}^N_x - \mathbf{q}^N_x \|_p^p\\
        & \Rightarrow  d^\text{cond}_p(pr, q)  \leq \| \mathbf{P}_C^+ \|_p \cdot  d^\text{pos}_p(pr, q) 
    \end{align*}
    }}

\begin{lemma}
\label{lemma:telescope}
    The $\ell_1$ distance between two sequence conditional distributions is bounded by the following inequality:\\
    \scalebox{0.92}{\parbox{1.08\linewidth}{
    \begin{align*}
        &|pr(x_1^N|c_1^N) - q(x_1^N|c_1^N)| = \left| \prod_{n=1}^N pr(x_n|c_n) - \prod_{n=1}^N q(x_n|c_n) \right|\\
        & =  \left |\sum_{j=1}^N \prod_{n=1}^{j-1} pr(x_n|c_n) \prod_{k=j+1}^N q(x_k|c_k) \cdot\big( pr(x_j|c_j) - q(x_j|c_j) \big)\right|\\
        &\leq  \sum_{j=1}^N \prod_{n=1}^{j-1} pr(x_n|c_n) \prod_{k=j+1}^N q(x_k|c_k) \cdot \left |pr(x_j|c_j) - q(x_j|c_j) \right|
    \end{align*}
    }}
    
\end{lemma}
    The derivation utilizes a telescope sum, by replacing the term $pr(x_n|c_n)$ in the product with $q(x_n|c_n)$ at each step.
\textbf{Proof of Theorem \ref{theorem:seqineq}}\\
 \scalebox{0.94}{\parbox{1.06\linewidth}{
\begin{align*}
     &\overline{D}_q \underbrace{\leq \sum_{x_1^N,c_1^N} |pr(c_1^N, x_1^N) - q(c_1^N, x_1^N)|}_{\text{triangular inequality}}\\
     &= \sum_{c_1^N} pr(c_1^N) \sum_{x_1^N}  \underbrace{\left| \prod_{n=1}^N pr(x_n|c_n) - \prod_{n=1}^N q(x_n|c_n)\right|}_{\text{apply Lemma \ref{lemma:telescope}}} \notag \\
    & \leq \sum_{c_1^N} pr(c_1^N)\sum_{x_1^N} \sum_{j=1}^N \prod_{n=1}^{j-1} pr(x_n|c_n) \cdot \\
    &\qquad \prod_{k=j+1}^N q(x_k|c_k) |pr(x_j|c_j) - q(x_j|c_j)|  \notag \\
    & =  \sum_{j=1}^N \sum_{c_1^N} pr(c_1^N)\sum_x |pr(x|c_j) - q(x|c_j)| \\
    & =  \sum_{j=1}^N \sum_{c} pr_j(c)\sum_x |pr(x|c) - q(x|c)|\\
    & \leq \sum_{j=1}^N \max_c pr_j(c) \underbrace{\sum_c\sum_x |pr(x|c) -q(x|c)|}_{\leq \|\mathbf{P}_C^{+}\|_1  d^\text{pos}_1(pr, q) 
    \text{ c.f. Lemma \ref{lemma:local}}}\\
    & \leq \underbrace{\sum_{j=1}^N \max_c pr_j(c)}_{\leq N} \|\mathbf{P}_C^{+}\|_1  \sum_n\underbrace{\sum_{x}|pr_n(x) -q_n(x)|}_{\leq \sum_{x_1^N}|pr(x_1^N) - q(x_1^N)|} \\
    &\leq N^2  \|\mathbf{P}_C^{+}\|_1  \sum_{x_1^N}|pr(x_1^N) - q(x_1^N)|
\end{align*}
}}
\begin{figure}
    \centering
    \includegraphics[scale=0.2]{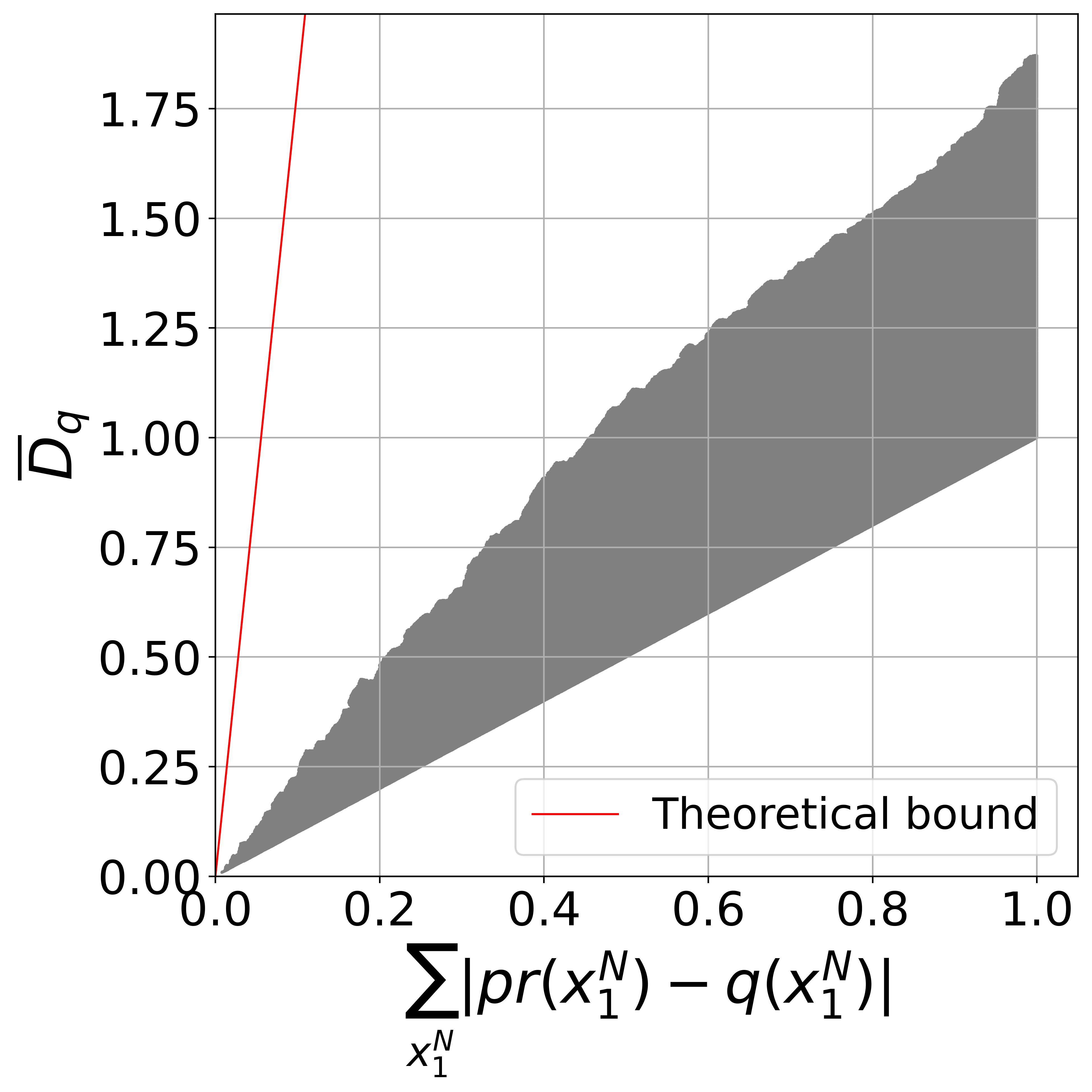}
    \caption{Simulation result for sequence-level bound between sequence-level marginal distributions $\sum_{x_1^N}|pr(x_1^N) - q(x_1^N)|$ and $\overline{D}_q$. The simulation is done with $|\mathcal{X}|=4$, $|\mathcal{C}|=3$, $N=3$, $\|\mathbf{P}_C^+\|_1 \leq 2$. $\mathbf{P}_C$ is guaranteed to have full column rank by conditioning that the minimum singular value is larger than 0.01. The grey dots refer to simulation points.}
    \label{fig:seq}
    \vspace{-0.2em}
\end{figure}

Figure \ref{fig:seq} shows the simulation results for Theorem \ref{theorem:seqineq}. The simulation was conducted by generating various distribution pairs $(pr,q)$ to cover the reachable areas as much as possible. The simulation result verifies the correctness of the bound.
\subsection{Sequence-Level Unsupervised Training Criterion}
By combining Ineq. \eqref{eq:basic}, \eqref{eq:global_seq} and applying Pinsker's inequality, we obtain the following bound between $\overline{\Delta}_q$ and the KL divergence $D_\text{KL}\big(pr(x_1^N)\|q(x_1^N)\big)$:\\
\scalebox{0.94}{\parbox{1.06\linewidth}{
\begin{align*}
    \big(\overline{\Delta}_q\big)^2 \leq \beta D_\text{KL}\big(pr(x_1^N)\|q(x_1^N)\big),
\end{align*}
}}
where $\beta = 2 N^4  \|\mathbf{P}_C^{+}\|^2_1$ is a fixed factor for a given training data, because the true distribution $pr$ is then fixed.
Therefore, minimizing $D_\text{KL}\big(pr(x_1^N)\|q(x_1^N)\big)$ also minimizes $\overline{\Delta}_q$. When $pr$ is fixed, minimizing the KL-divergence is equivalent to minimize the cross-entropy $H(pr(x_1^N), q(x_1^N))$. In practice, the true distribution $pr(x_1^N)$ is approximated by the empirical distribution $\Tilde{pr}(x_1^N) =\frac{1}{S} \sum_{s=1}^S \delta(x_1^N, x_{s,1}^N)$,
where $\{x_{s,1}^N\}_{s=1}^S$ are sequences in the unsupervised training data \cite{ney2003relationship, yang2025classification}. The sequence-level cross-entropy training criterion is:\\
\scalebox{0.93}{\parbox{1.07\linewidth}{
\vspace{-0.1em}
\begin{align*}
    \mathcal{L}(\theta) &=-\sum_{x_1^N} \Tilde{pr}(x_1^N) \log q(x_1^N)= - \frac{1}{S}\sum_{s=1}^S \log q_\theta(x_{s,1}^N) \\
    &= - \frac{1}{S}\sum_{s=1}^S \log \sum_{c_1^N} p_\text{LM}(c_1^N) q_\theta(x_{s,1}^N|c_1^N)
\end{align*}
}}
where $p_\text{LM}(c_1^N)$ is a language model (LM) on trained text data to approximate $pr(c_1^N)$. The summation over $c_1^N$ can be computed efficiently via dynamic programming for a limited context LM \cite{povey2016purely}; for a full-context LM, one can restrict the summation to a hypothesis space obtained by search. For discriminative models that output a posterior $q(c|x)$, it is possible to apply Bayes rule $q(x|c) = q(c|x) pr(x)/q(c)$ to obtain $q(x|c)$, similar to a HMM model. Overall, this loss enables one-stage optimization of the statistical model.

\section{Discussions}
In this section, we will show that when no further constraints are imposed, the two conditions we propose are also necessary for unsupervised speech recognition. To this end, we consider the exact-match case $pr(x_1^N)= q(x_1^N)$ and show that $\Delta_q>0 $ can still occur via explicit construction of distributions if either of these two conditions is violated.

\subsection{Necessity of the Full-Column Rank Condition}
\label{sec:necessary_full_column}
Consider a position-dependent unigram LM $pr(c_1^N) = \prod_n pr_n(c_n)$. Then, the probability $pr(x_1^N)$ can also be decomposed per position: $pr(x_1^N) = \prod_n pr_n(x_n)$, so as the model distribution $q(x_1^N)$. Due to the decomposition of $pr(c_1^N)$, $pr(x_1^N) \equiv q(x_1^N) \quad \text{iff } \forall n, pr_n(x_n) \equiv q_n(x_n)$. Then, finding the appropriate $q(x|c)$ becomes solving linear equations in \eqref{eq:loca}:\\
\scalebox{0.93}{\parbox{1.07\linewidth}{
\begin{align*}
    \mathbf{P}_C \big( \mathbf{pr}^C_{x} - \mathbf{q}^C_{x}\big) = \mathbf{0}
\end{align*}
}}
 When $\text{rank}(\mathbf{P}_C) =r \leq |\mathcal{C}|-1$, there are effectively in total $r\cdot |\mathcal{X}|$ linear equations. Considering the normalization of $q(x|c)$, there is in total $r\cdot   |\mathcal{X}| + |\mathcal{C}| $ number of linear equations, and $|\mathcal{C}||\mathcal{X}|$ number of variables to solve for $q(x|c)$. Since we assume $|\mathcal{X}|> |\mathcal{C}|$, we have $r|\mathcal{X}|+|\mathcal{C}|< |\mathcal{X}||\mathcal{C}|$. Therefore, the system is underdetermined. Moreover, with a position-dependent unigram LM, the decision rule also becomes local, which makes the problem similar to the single event case in Sec. \ref{sec:problem}, i.e. $\overline{\Delta}_q>0$ is possible. 

\subsection{Necessity of the Structure Assumption}
Assume $pr(c_1^N, x_1^N) = \prod_n pr_n(c_n,x_n)$. As in Sec. \ref{sec:necessary_full_column}, the exact-match case localizes to each position $n$:\\
\scalebox{0.93}{\parbox{1.07\linewidth}{
\begin{align*}
    \mathbf{P}_C \mathbf{q}^C_{x} = \mathbf{pr}^N_{x}
\end{align*}
}}
In general, although $\mathbf{P}_C^+$ exists under the full-column-rank assumption, a solution of $\mathbf{q}^C_{x}$ is not guaranteed. Even if there is a solution, which is then a unique one, a bound for error mismatch is not guaranteed. Let $\hat{q}(x|c)$ be the corresponding unique solution, for each $x$ and $n$:\\
\scalebox{0.94}{\parbox{1.06\linewidth}{
\begin{equation*}
\label{eq:structure_n}
         \sum_c pr_n(c) pr_n(x|c) = pr_n(x) = q_n(x) = \sum_c pr_n(c) \hat{q}(x|c)
\end{equation*}
}}
Since $pr_n(x|c)$ can be different for different $n$, the true distributions $\{pr_n(x|c)\}$ that satisfy the equation above is not unique. Therefore, for each position, it is similar to the single event case and $\overline{\Delta}_q>0$ is possible.

 \section{Conclusion}
We developed a theoretical framework for sequence-level unsupervised training in automatic speech recognition. We introduced two conditions under which unsupervised speech recognition is possible, and showed that, without further additional constraints, these conditions are also necessary. Under these two conditions, we further derived a classification error bound for unsupervised speech recognition. Guided by this bound, we propose a sequence-level cross-entropy loss for end-to-end unsupervised training.
{
\begin{spacing}{0.5}
{\footnotesize
\unnumberedsection{Acknowledgments}


\scriptsize
\selectfont
This work was partially supported by NeuroSys, which as part of the
initiative “Clusters4Future” is funded by the Federal Ministry of
Education and Research BMBF (funding IDs 03ZU2106DA).
}
\end{spacing}
}





\bibliographystyle{IEEEbib}
\bibliography{strings,refs}

\end{document}